\documentclass[letterpaper, 10pt, onecolumn,conference]{ieeeconf}  
\IEEEoverridecommandlockouts                              

\usepackage{amsmath}
\usepackage{amssymb}
\usepackage{amsfonts}
\usepackage{mathptmx} 
\usepackage{times}
\usepackage{graphicx}
\usepackage{color}

\usepackage{hyperref}
\DeclareMathAlphabet{\bit}{OML}{cmm}{b}{it}

\newtheorem{thm}{Theorem}

\usepackage{datetime}
\usepackage{times} 
\usepackage{framed} 
\usepackage{graphicx} 
\usepackage{amsmath} 
\usepackage{amssymb}  

\def\<{\leqslant}           
\def\>{\geqslant}           

\def\d{\partial}
\def\wh{\widehat}

\def\mR{\mathbb{R}}    

\def\Tr{\mathrm{Tr}}       
\def\rT{\mathrm{T}}        
\def\bS{\mathbf{S}}

\def\bE{\mathbf{E}}    

\def\[[[{[\![\![}
\def\]]]{]\!]\!]}

\def\bra{{\langle}}
\def\ket{{\rangle}}

\def\Bra{\left\langle}
\def\Ket{\right\rangle}

\def\re{\mathrm{e}}        
\def\rd{\mathrm{d}}        

\def\x{\times}

\def\bT{\mathbf{T}}

\def\cF{\mathcal{F}}

\def\cC{\mathcal{C}}

\def\mS{\mathbb{S}}

\def\Ups{\Upsilon}
\def\ups{\upsilon}

\title{\LARGE \bf 
Mean Square Optimal Control by Interconnection for 
Linear Stochastic Hamiltonian Systems$^*$}

\usepackage{datetime}

\author{Igor G. Vladimirov$^{\dagger}$, \qquad Ian R. Petersen$^{\dagger}$\thanks{$^*$This work is supported by the Australian Research Council under grant DP160101121.}
\thanks{$^\dagger$Research School of Electrical, Energy and Materials Engineering, College of Engineering and Computer Science, Australian National University,
Canberra, Acton, ACT 2601, Australia,
{\tt igor.g.vladimirov@gmail.com, i.r.petersen@gmail.com}.
}
}

\begin{document}
\maketitle
\thispagestyle{empty}
\pagestyle{plain}
\begin{abstract}
This paper is concerned with linear stochastic Hamiltonian (LSH) systems subject to random external forces. Their dynamics are modelled by linear stochastic differential equations, parameterised by stiffness, mass, damping and coupling matrices.   A class of physical couplings  is discussed for such systems using inerters, springs and dampers.  We consider a problem of minimising a steady-state quadratic cost functional over the coupling parameters for the interconnection of two LSH systems, one of which plays the role of an analog controller.
 For this mean square optimal control-by-interconnection setting, we outline first-order necessary conditions of optimality which  employ variational methods developed previously for constrained linear quadratic Gaussian control problems.


\end{abstract}

\section{\bf Introduction}

Physical aspects of mechanical and electrical systems, arising in engineering applications, are usually represented in terms of energy functions,  captured in the Hamiltonian or Lagrangian,  and constitutive relations for nonconservative elements. The energy conservation and dissipation play an important role  in the dynamics of Hamiltonian systems and their open versions, port-Hamiltonian systems \cite{OVMM_2001,OVME_2002,VJ_2014}, which interact between themselves and with the surroundings. The energy transfer  mechanisms are employed in the control-by-interconnection approach to achieving performance specifications for networks of such systems,  whose analysis uses dissipation inequalities with storage and supply rate functions \cite{SVP_2020,V_2016,W_1972}.

In comparison with controllers, which can produce a fairly arbitrary  actuator signal based on digital processing of measurements, a controller in the form of a physical system (such as used, for example, in centrifugal governors \cite{M_1868})  is constrained by the Hamiltonian structure. If a physical system interacts with a complex environment, which does not lend itself to a deterministic description, the Hamiltonian structure is combined with a stochastic model of the external random forcing. Examples include the vehicle suspension affected by an uneven road profile, elastic structures in turbulent fluid flows, and  electrical circuits subject to thermal noise. Such applications lead to stochastic optimal control problems with Hamiltonian structure constraints, similar to those arising in coherent quantum control \cite{JG_2010,NJP_2009} from physical realizability conditions \cite{JNP_2008,SP_2012}.

The present paper is concerned with a model class of linear stochastic Hamiltonian (LSH) systems \cite{VP_2018b} subject to random external forces. The  evolution of such a system is governed by a linear stochastic differential equation (SDE), driven by an Ito process and parameterised by a quadruple of stiffness, mass, damping and coupling matrices. These parameters specify the energetics of the system (including the quadratic Hamiltonian and Langevin viscous damping) and its interaction with the environment. We discuss a multivariable stochastic version  of physical couplings for such systems, which involve inerters \cite{S_2002,S_2003}, springs and dampers.  In the resulting interconnection of two LSH systems, one of them is interpreted as a plant, while the other plays the role of an analog (rather than digital) controller. The infinite-horizon performance of the interconnected LSH system, driven by a standard  Wiener process, is quantified by a steady-state mean square cost functional. This gives rise to a linear-quadratic-Gaussian (LQG) control problem of minimising the cost over the coupling parameters within the  Hamiltonian structure of the interconnection. For this mean square optimal control-by-interconnection problem, we outline first-order necessary conditions of optimality which  employ variational techniques (such as Frechet differentiation in matrix-valued variables) developed previously for constrained LQG control problems \cite{BH_1998,SIG_1998,VP_2013a}.

The paper is organised as follows.
Section~\ref{sec:LSH} specifies the class of LSH systems and reviews their properties such as 
the energy balance relations and stability.  
Section~\ref{sec:link} discusses a multivariable stochastic version of the inerter-spring-damper coupling of LSH systems.
Section~\ref{sec:msopt} considers a mean square optimal control problem for the LSH system interconnection and outlines first-order necessary conditions of optimality.
Section~\ref{sec:conc} provides concluding remarks.

\section{\bf Linear stochastic Hamiltonian systems}
\label{sec:LSH}

 Consider a linear stochastic Hamiltonian (LSH) system \cite{VP_2018b} with $n$ degrees of freedom. The system is endowed  with $\mR^n$-valued vectors  $q:= (q_k)_{1\< k\< n}$, $\dot{q} = (\dot{q}_k)_{1\< k\< n}$ and $p:= (p_k)_{1\< k\< n}$  of generalised coordinates (positions), velocities and momenta, where $\dot{(\ )}:= \d_t$ is the time derivative, and vectors are organised as columns unless specified otherwise. The velocities and momenta are related by
 \begin{equation}
\label{pqdot}
    p
    :=
    \d_{\dot{q}} T
    =
    M\dot{q},
\end{equation}
where $M$ is a real positive definite symmetric \emph{mass matrix} of order $n$.  Here,
 \begin{equation}
 \label{kin}
    T(p):=
    \frac{1}{2}
    \|\dot{q}\|_M^2
    =
    \frac{1}{2} \|p\|_{M^{-1}}^2
 \end{equation}
is the kinetic energy of the system, with $\|v\|_E:= \sqrt{v^{\rT}E v} = |\sqrt{E}v|$ a weighted Euclidean (semi-) norm of a vector $v$ specified by a positive (semi-) definite matrix $E$.  In the case of  rotational degrees of freedom, the generalised coordinates $q_1, \ldots, q_n$ are angular positions, the mass matrix $M$ corresponds to the tensor of inertia, and $p$ is the angular momentum vector. In the context of electrical networks (such as RLC circuits),   the above quantities are interpreted according to electromechanical analogies.  The potential energy of the LSH system is a quadratic  form of the position vector:
\begin{equation}
\label{VK}
    V(q)
    =
    \frac{1}{2}
    q^\rT K q,
\end{equation}
where $K=K^\rT \in \mR^{n\x n}$ is a \emph{stiffness matrix}.
The system Hamiltonian  $H: \mR^{2n} \to \mR$ on the phase space $\mR^{2n} = \mR^n \x \mR^n$ (the product of the position and momentum spaces) is the sum of
the potential energy (\ref{VK}) and the kinetic energy (\ref{kin}):
\begin{equation}
\label{HR}
    H(q,p)
    :=
    V(q)
    +
    T(p)
    =
  \frac{1}{2}
  x^\rT R x,
  \qquad
    x
    :=
    {\begin{bmatrix}
        q\\
        p
    \end{bmatrix}},
\end{equation}
where
\begin{equation}
\label{R}
  R
  :=
  {\begin{bmatrix}
    K & 0\\
    0 & M^{-1}
  \end{bmatrix}}
\end{equation}
is the \emph{energy matrix}. The system is also endowed with an $\mR^m$-valued output $y:=(y_k)_{1\< k\< m}$. The position $q$, the momentum $p$, and the output $y$    of the LSH system  evolve in time according to the equations
\begin{align}
\label{LSH1}
    \dot{q}
    &= \d_pH =  M^{-1}p,\\
\nonumber
    \rd p
    & =
    (-\d_q H - F\dot{q} )\rd t
    +
    N^\rT\rd W\\
\label{LSH2}
    & =
    (-Kq - FM^{-1}p)\rd t
        +
          N^\rT
        \rd W,\\
\label{LSH3}
    y & = N q,
\end{align}
where the ODE (\ref{LSH1}) is obtained from (\ref{pqdot}).  The SDE (\ref{LSH2}), which is equivalent to \begin{equation}
\label{LSH2q}
    M\rd \dot{q}= (-Kq-F\dot{q})\rd t + N^\rT \rd W
\end{equation}
(with $\rd \dot{q}:= \rd (\dot{q})$),
 is driven by an $\mR^m$-valued  random process $W:=(W_k)_{1\< k\< m}$ which is assumed to be an Ito process \cite{RW_2000} 
 with respect to an underlying  filtration $\cF:= (\cF_t)_{t\> 0}$.
The process $W$, whose structure is specified below,
models an external input random forcing on the system, with $N^\rT W$ having the physical dimension of momentum in accordance with (\ref{LSH2}).
The dispersion matrix $N^\rT$ of this SDE is specified by a system-environment \emph{coupling matrix} $N\in \mR^{m\x n}$ which relates the output $y$ to the position $q$ in (\ref{LSH3}). The term $-F\dot{q} = -FM^{-1}p$ in (\ref{LSH2}) is  the Langevin viscous damping force, specified by   a \emph{damping matrix} $F \in \mS_n^+$, where $\mS_n^+$ denotes the set of real positive semi-definite symmetric matrices  of order $n$.
In terms of the $\mR^{2n}$-valued state process $x$ in (\ref{HR}), the equations of motion (\ref{LSH1})--(\ref{LSH3}) are represented  as
\begin{align}
\label{dXlin}
    \rd x
    & =
    \Big(
        J
        -
        {\begin{bmatrix}
          0 & 0 \\
          0 & F
        \end{bmatrix}}
    \Big)
    H'
        \rd  t
        +
        {\begin{bmatrix}
          0 \\
          N^\rT
        \end{bmatrix}}\rd W
        =
    A x \rd t + B\rd W,\\
\label{yCx}
    y &
    =
    Cx,
\end{align}
with appropriately dimensioned state-space matrices $A$, $B$, $C$ given by 
\begin{align}
\label{A}
    A
    & :=
        \Big(
        J
        -
        {\begin{bmatrix}
          0 & 0 \\
          0 & F
        \end{bmatrix}}
    \Big)
    R
     =
        {\begin{bmatrix}
          0 & M^{-1} \\
          -K & -FM^{-1}
        \end{bmatrix}},\\
\label{B}
        B
        & :=
        {\begin{bmatrix}
          0 \\
          N^{\rT}
        \end{bmatrix}},\\
\label{C}
        C
        & :=
        {\begin{bmatrix}
          N & 0
        \end{bmatrix}},
\end{align}
cf. \cite[Eq. (20)]{V_2016}. Here, 
$    J:=
    {\small\begin{bmatrix}
        0 & I_n \\
        -I_n & 0 \\
    \end{bmatrix}}$ 
(with $I_n$ the identity matrix of order $n$) is the symplectic structure matrix
which generates the Poisson bracket    \cite{A_1989}
\begin{equation}
\label{Poiss}
    \{f, g\}
    :=
    f'^{\rT} J g'
    =
    \d_q f^{\rT}\d_p g - \d_p f^{\rT}\d_q g
\end{equation}
for smooth functions $f,g:\mR^{2n}\to \mR$ on the phase space.
Also,
$(\cdot)'$ denotes the gradient of a function with respect to all its variables, so that
\begin{equation}
\label{H'}
    H'
    =
    {\begin{bmatrix}
        \d_q H \\
        \d_p H
    \end{bmatrix}}
    =
    {\begin{bmatrix}
        \d_q V \\
        \d_p T
    \end{bmatrix}}
    =
    {\begin{bmatrix}
        K q \\
        M^{-1} p
    \end{bmatrix}}
    =
    Rx
\end{equation}
in (\ref{dXlin})
consists of the gradients of the Hamiltonian over the positions and momenta. Accordingly,
$-V' = -Kq$ is a potential force field.
The LSH system in (\ref{LSH1})--(\ref{LSH3})  (or (\ref{dXlin})--(\ref{C}))  is parameterised by the quadruple $(K,M,F,N)$ of the stiffness, mass, damping and coupling matrices. 
If, in addition to $M\succ 0$, the matrices $K$ and $F$ are also positive definite, then the matrix $A$ in (\ref{A}) is Hurwitz \cite[Theorem 1]{VP_2018b} and the LSH system is internally stable. Although the stiffness matrix $K$ is positive definite in standard mass-spring-dashpot models,  negative effective stiffness can also be achieved by using special mechanical arrangements \cite{LGT_2007}. The static memoryless dependence of the output $y$ on the internal position vector $q$ in (\ref{LSH3}) also admits mechanical implementation, for example, by using levers or more complicated linkages; see Fig.~\ref{fig:KMFN}.
\begin{figure}[htbp]
\centering
\includegraphics[width=7cm]{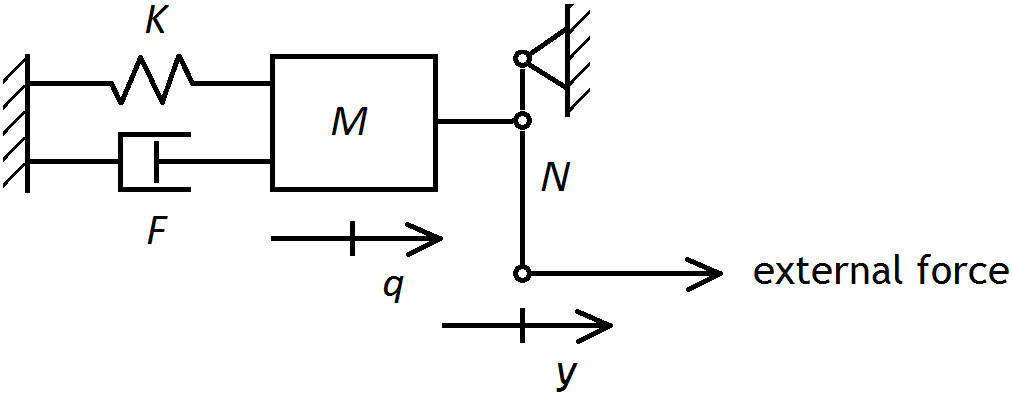}\vskip-1mm
  \caption{A mass-spring-dashpot system of mass $M$ with spring of stiffness $K$ and dashpot  of damping constant $F$. The internal position $q$ is the deviation of the mass from its equilibrium.  The external force is exerted on the free end (with the output displacement $y$) of a frictionless lever which implements the system-environment coupling $N$. In this example, $N>0$ is the mechanical advantage of the lever. All the elements of the system, except for the mass $M$, are assumed massless.
    }
    \label{fig:KMFN}
\end{figure}

The incremental work $\dot{y}^\rT \rd W$ of the external force on the system contributes to its energy balance \cite{VP_2018b} in terms of the stochastic differential of the Hamiltonian (\ref{HR}) as
\begin{align}
\nonumber
    \rd  H
    & =
    H'^{\rT}\rd x
    +
    \frac{1}{2}
    \bra
        \d_p^2 H,
        N^\rT \Sigma N
    \ket
    \rd t\\
\nonumber
    &=
    \Big(
        \{H,H\}
        -
        \d_p H^\rT F\dot{q}
        +
        \frac{1}{2}
        \bra
            M^{-1},
            N^\rT \Sigma N
        \ket
    \Big)\rd t
        +
        \d_p H^{\rT}
        N^\rT \rd W\\
\label{dH}
        &=
        \Big(
            -
            \|\dot{q}\|_F^2
    +
            \frac{1}{2}
            \bra
                N M^{-1}N^\rT,
                \Sigma
            \ket
        \Big)\rd t
        +
        \dot{y}^\rT
        \rd W,
\end{align}
where $\bra a,b\ket:= \Tr(a^\rT b)$ is the Frobenius inner product of real matrices, with $\|a\|:= \sqrt{\bra a,a \ket}$ the corresponding norm. The potential energy $V$ does not enter the right-hand side of (\ref{dH}) since $\{H,H\} =0$ due to the antisymmetry of the Poisson bracket (\ref{Poiss}).   The energy dissipation rate
$    -\dot{q}^{\rT}F\dot{q}
    =
    -\|\dot{q}\|_F^2
    \< 0
$
in the drift of (\ref{dH}) comes from the work of the damping force  over the system. The additional term $            \frac{1}{2}
            \bra
                N M^{-1}N^\rT,
                \Sigma
            \ket\>0
$ originates from the Ito lemma \cite{RW_2000} 
combined with (\ref{HR})--(\ref{yCx}), (\ref{H'}), the  Hessian
$
    \d_p^2H = \d_p^2T = M^{-1}
$ of the kinetic energy and the diffusion matrix $N^\rT \Sigma N$ of the process $N^\rT W$. The latter matrix is related to
an $\mS_m^+$-valued process $\Sigma$, adapted to the filtration $\cF$
and specifying the Ito table
\begin{equation}
\label{dWdW}
    \rd W \rd W^\rT = \Sigma \rd t,
    \qquad
    \Sigma := \beta\beta^\rT
\end{equation}
for the Ito process $W$, which drives the SDE (\ref{LSH2}) and has the stochastic differential
\begin{equation}
\label{dW}
  \rd W(t) = \alpha(t) \rd t + \beta(t) \rd \ups(t).
\end{equation}
Here, $\alpha$, $\beta$ are $\cF$-adapted processes with values in $\mR^m$, $\mR^{m\x m}$,  respectively, satisfying 
$    \int_0^t
    (|\alpha(\tau)| + \|\beta(\tau)\|^2)
    \rd \tau
    <
    +\infty
$ 
almost surely
for any time $t>0$, and $\ups$ is a standard Wiener process \cite{RW_2000} 
in $\mR^m$ with respect to the filtration $\cF$. Since $\|\beta\|^2 = \Tr \Sigma$ in view of (\ref{dWdW}), the term
$    
    \int_0^t
    \Tr \Sigma(\tau)
    \rd \tau
$ 
describes the quadratic variation 
of the process $W$ in (\ref{dW}) over the time interval $[0,t]$. In the absence of diffusion, when $\Sigma(t) =  0$ for all $t\> 0$, the momentum $p$  is an absolutely continuous  function of time, and the SDEs (\ref{LSH2}), (\ref{LSH2q}) reduce to the ODE
$
    \dot{p}
    =
    M\ddot{q}
    =-Kq - F\dot{q}
        +
          N^\rT \alpha
$
driven by the external force $N^\rT \alpha$. 
Returning to the general setting, note that the special structure (\ref{A})--(\ref{C})   of the LSH system (as an input-output operator $W\mapsto y$)  manifests itself in the form of the transfer function  $$
    \Phi(s):= \chi(s)B = N(K + sF + s^2 M)^{-1}N^\rT,
$$
which, similarly to the case of deterministic linear time-invariant systems, relates the Laplace transforms
$$
    {\begin{bmatrix}
      \wh{x}(s) &
      \wh{y}(s) &
      \wh{W}(s)
    \end{bmatrix}}
    :=
    \int_{0}^{+\infty}
    \re^{-st}
    {\begin{bmatrix}
      x(t)\rd t &
      y(t)\rd t &
      \rd W(t)
    \end{bmatrix}}
$$
of the processes in (\ref{dXlin}), (\ref{yCx})
(where $\wh{W}$ is associated with the incremented process $W$) as
$$
    \wh{y}(s) = C\wh{x}(s) = \Phi(s)\wh{W}(s) + \chi(s)x(0),
    \quad
    \chi(s):= C(s I_{2n}-A)^{-1}.
$$
Accordingly, the static gain matrix of the system is symmetric:
\begin{equation}
\label{T0}
    \Phi(0)
    =
    NK^{-1}N^\rT,
\end{equation}
provided the stiffness matrix $K$ is nonsingular. In the case of $K\succ 0$, (\ref{T0}) yields $\Phi(0)\succcurlyeq 0$, which reflects the property of  a usual spring under a static load to deform in the direction of the force applied.

\section{\bf Inerter-spring-damper coupling of LSH systems}
\label{sec:link}

The input force acting on the LSH system can come from its relative position, velocity and acceleration   with respect to another such system or an external reference position signal (for example, an uneven road profile \cite{BPR_2012} affecting the car suspension). A fairly general class of physical couplings, which, similarly to the deterministic case of \cite{P_2015},  convert the position variables and their time derivatives (or appropriate Ito increments) into forces, is provided by the following multivariable stochastic version of \emph{inerter-spring-damper}   links \cite{S_2002}. Such a coupling has two terminals whose $\mR^m$-valued positions $y_1$, $y_2$ (measured in one direction) are assumed to be absolutely continuous random functions of time, with $\dot{y}_1$, $\dot{y}_2$ being Ito processes which give rise to the forcing $\omega$ at the first terminal as
\begin{equation}
\label{f12}
    \rd \omega =  (\kappa(y_2-y_1)+\phi(\dot{y}_2-\dot{y}_1))\rd t + \mu \rd (\dot{y}_2-\dot{y}_1).
\end{equation}
Here, $\mu\succcurlyeq 0$, $\kappa$, $\phi\succcurlyeq 0$ are real symmetric \emph{inertance} \cite{S_2002}, stiffness and damping  matrices of order $m$. Under the condition that the $(\mu,\kappa,\phi)$-coupling  is massless,  the forcing at the second terminal of the link is  equal in magnitude and opposite in direction due to Newton's third law:
\begin{equation}
\label{f21}
    -\rd \omega =(\kappa(y_1-y_2)+\phi(\dot{y}_1-\dot{y}_2))\rd t + \mu \rd (\dot{y}_1-\dot{y}_2).
\end{equation}
Although the role of $\mu$ is similar to that of the mass matrix, the idealised model assumption on the link being massless is satisfied  to a high degree of accuracy by making the actual mass of the inerter negligible \cite{S_2003} compared to its inertance.

As an example of interconnection of two LSH systems $(K_1,M_1,F_1,N_1)$, $(K_2,M_2,F_2,N_2)$ (each with the same number $n$ of degrees of freedom, for simplicity), suppose their outputs $y_1$, $y_2$ are coupled through the $(\mu,\kappa,\phi)$-link and are subject to the external $\mR^m$-valued forcing $W_1$, $W_2$ as shown in Fig.~\ref{fig:link}.
\begin{figure}[htbp]
\centering
\includegraphics[width=10cm]{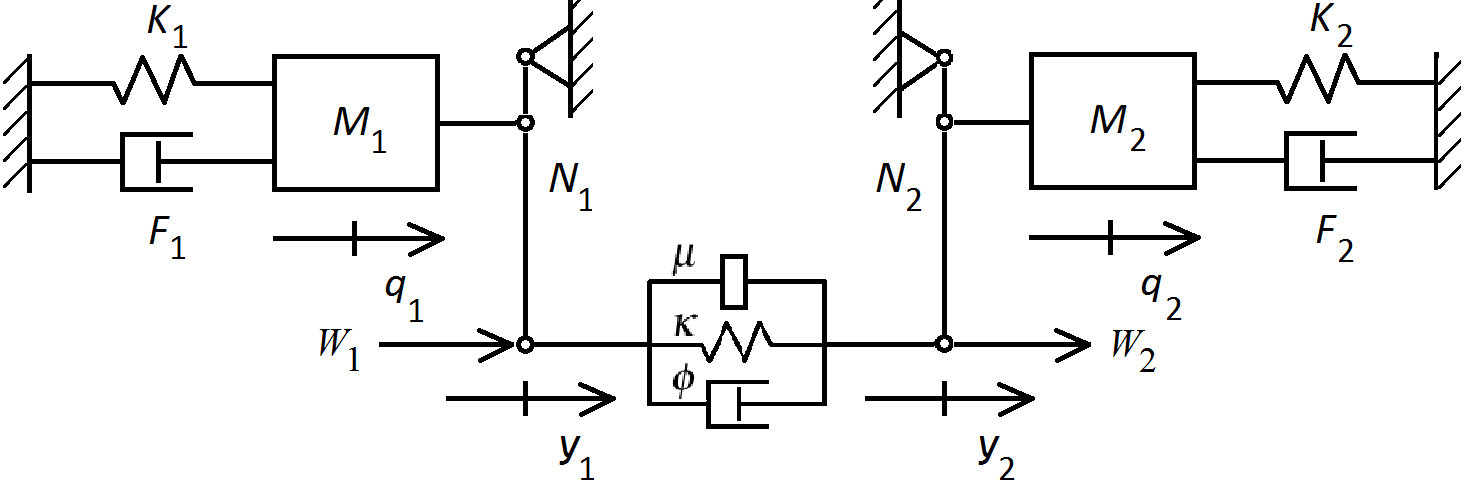}\vskip-1mm
  \caption{Two LSH systems $(K_1,M_1,F_1,N_1)$ and $(K_2,M_2,F_2,N_2)$ connected through an ideal inerter-spring-damper $(\mu,\kappa,\phi)$-coupling and subjected to external forcing $W_1$,  $W_2$. Here, $q_1$, $q_2$ are the deviations of the systems  from their equilibrium positions, and $y_1$, $y_2$ are the corresponding outputs.}
    \label{fig:link}
\end{figure}
The resulting system is governed by
\begin{align}
\label{qk}
    \dot{q}_k
    & =
    M_k^{-1} p_k,\\
\label{pk}
    \rd p_k
    & =
    (-K_kq_k - F_k\dot{q}_k)\rd t
        +
          N_k^\rT
        (\rd W_k-(-1)^k\rd \omega),\\
\label{yk}
    y_k &
    =
      N_kq_k,
      \qquad
      k = 1,2,
\end{align}
where the forcing (\ref{f12}) on the first terminal of the link takes the form
\begin{equation}
\label{dom}
    \rd \omega
    =
    (\kappa Z q+\phi Z\dot{q})\rd t
    + \mu Z \rd \dot{q},
    \qquad
        Z:=
    {\begin{bmatrix}
      -N_1 & N_2
    \end{bmatrix}}.
\end{equation}
Here,
\begin{equation}
\label{q}
    q:=
    {\begin{bmatrix}
      q_1\\
      q_2
    \end{bmatrix}}
\end{equation}
is formed from the $\mR^n$-valued positions $q_1$, $q_2$ of the LSH systems.
Note that $p_1$, $p_2$ in (\ref{qk}), (\ref{pk}) are the \emph{individual} momenta of the systems, which, unlike the positions,  are  not merely subvectors of the momentum  of the interconnection whose structure is discussed below.

\begin{thm}
\label{th:link}
The interconnection of the LSH systems $(K_1,M_1,F_1,N_1)$, $(K_2,M_2,F_2,N_2)$ through the $(\mu,\kappa,\phi)$-coupling in Fig.~\ref{fig:link}, described by (\ref{qk})--(\ref{dom}), is an LSH $(K,M,F,N)$-system whose augmented position $q$, input $W$ and output $y$ are given by (\ref{q}) and
\begin{equation}
\label{Wy}
    W
    :=
    {\begin{bmatrix}
        W_1\\
        W_2
    \end{bmatrix}},
    \qquad
    y
    :=
    {\begin{bmatrix}
        y_1\\
        y_2
    \end{bmatrix}}
    =
    Nq.
\end{equation}
The parameters of the interconnected system are computed as
\begin{align}
\label{K}
    K
    & :=
    K_0
    +
    Z^\rT \kappa Z,\\
\label{M}
    M
    & :=
    M_0
    +
    Z^\rT \mu Z, \\
\label{F}
    F
    & :=
    F_0
    +
    Z^\rT \phi Z,\\
\label{N}
    N
    & :=
    {\begin{bmatrix}
      N_1 & 0\\
      0 & N_2
    \end{bmatrix}}
\end{align}
using the block diagonal matrices
\begin{equation}
\label{KMF0}
    K_0:=
    {\begin{bmatrix}
      K_1 & 0\\
      0 & K_2
    \end{bmatrix}},
    \qquad
    M_0:=
    {\begin{bmatrix}
      M_1 & 0\\
      0 & M_2
    \end{bmatrix}},
    \qquad
    F_0:=
    {\begin{bmatrix}
      F_1 & 0\\
      0 & F_2
    \end{bmatrix}},
\end{equation}
associated with the LSH systems as if they were decoupled.\hfill$\square$
\end{thm}
\begin{proof}
We will use auxiliary Ito processes $\pi_1$, $\pi_2$ with the stochastic differentials
\begin{equation}
\label{pik}
  \rd \pi_k
   :=
  (-K_kq_k - F_k\dot{q}_k)\rd t
        +
          N_k^\rT
        \rd W_k,
        \qquad
        k = 1,2.
\end{equation}
In view of (\ref{q}), (\ref{Wy}), (\ref{N}), (\ref{KMF0}), the relations (\ref{pik}) can be assembled into
\begin{equation}
\label{pipi}
    \rd
      {\begin{bmatrix}
        \pi_1\\
        \pi_2
    \end{bmatrix}}
    =
  (-K_0q - F_0\dot{q})\rd t
        +
          N^\rT
        \rd W.
\end{equation}
This allows
the SDEs (\ref{pk}) to be combined as
\begin{equation}
\label{pp}
    \rd
    {\begin{bmatrix}
        p_1\\
        p_2
    \end{bmatrix}}
    =
    \rd
    {\begin{bmatrix}
        \pi_1\\
        \pi_2
    \end{bmatrix}}
    -Z^\rT
    \rd \omega,
\end{equation}
where (\ref{pipi}) is used together with the matrix $Z$ from (\ref{dom}). Since
\begin{equation}
\label{qdot}
    \dot{q} =
    {\begin{bmatrix}
        M_1^{-1}p_1\\
        M_2^{-1}p_2
    \end{bmatrix}}
    =
    M_0^{-1}
    {\begin{bmatrix}
        p_1\\
        p_2
    \end{bmatrix}}
\end{equation}
in view of (\ref{qk}), (\ref{q}), (\ref{KMF0}), substitution of (\ref{pp}) into (\ref{dom}) leads to
\begin{align}
\nonumber
    \rd \omega
    &=
    \gamma \rd t
    +
    \mu
    Z M_0^{-1}
    \rd
    {\begin{bmatrix}
        p_1\\
        p_2
    \end{bmatrix}}\\
\label{dom1}
    &=
    \gamma \rd t
    +
    \mu
    Z
    M_0^{-1}
    \rd
    {\begin{bmatrix}
        \pi_1\\
        \pi_2
    \end{bmatrix}}
    -
    \mu
    Z
    M_0^{-1}
    Z^\rT
    \rd \omega,
\end{align}
where
\begin{equation}
\label{gamma}
    \gamma:= \kappa Z q+\phi Z\dot{q}
\end{equation}
is the drift term in the SDE
(\ref{dom}) which comes from the spring-damper  part of the coupling between the LSH systems; see Fig.~\ref{fig:link}. Now, (\ref{dom1}) is solved for $\rd \omega$ as
\begin{equation}
\label{dom2}
    \rd \omega
    =
    (I_m +     \mu ZM_0^{-1} Z^\rT)^{-1}
    \Big(
    \gamma \rd t
    +
    \mu
    Z
    M_0^{-1}
    \rd
    {\begin{bmatrix}
        \pi_1\\
        \pi_2
    \end{bmatrix}}
\Big).
\end{equation}
Here, the nonsingularity $\det(I_m +     \mu ZM_0^{-1} Z^\rT)\ne 0$ is secured by
$\mu ZM_0^{-1} Z^\rT$ being isospectral to the matrix $\sqrt{\mu} ZM_0^{-1} Z^\rT \sqrt{\mu}  \succcurlyeq 0$ since $\mu\succcurlyeq 0$ and $M_0\succ 0$.  Substitution of (\ref{dom2}) into (\ref{pp}) leads to
\begin{align}
\nonumber
    \rd
    {\begin{bmatrix}
        p_1\\
        p_2
    \end{bmatrix}}
    =&
    \rd
    {\begin{bmatrix}
        \pi_1\\
        \pi_2
    \end{bmatrix}}
    -Z^\rT
    (I_m +     \mu ZM_0^{-1} Z^\rT)^{-1}
    \Big(
    \gamma \rd t
    +
    \mu
    Z
    M_0^{-1}
    \rd
    {\begin{bmatrix}
        \pi_1\\
        \pi_2
    \end{bmatrix}}
\Big)\\
\nonumber
    =&-Z^\rT
    (I_m +     \mu ZM_0^{-1} Z^\rT)^{-1}
    \gamma \rd t\\
\nonumber
& +
    (I_{2n} - Z^\rT
    (I_m +     \mu ZM_0^{-1} Z^\rT)^{-1}    \mu
    Z
    M_0^{-1}
)
    \rd
    {\begin{bmatrix}
        \pi_1\\
        \pi_2
    \end{bmatrix}}\\
\nonumber
= &
    (I_m + Z^\rT \mu ZM_0^{-1})^{-1}
    \Big(
        -
    Z^\rT
    \gamma \rd t
    +
    \rd
    {\begin{bmatrix}
        \pi_1\\
        \pi_2
    \end{bmatrix}}
    \Big)\\
\label{pp1}
= &
    M_0M^{-1}
    \Big(
        -
    Z^\rT
    \gamma \rd t
    +
    \rd
    {\begin{bmatrix}
        \pi_1\\
        \pi_2
    \end{bmatrix}}
    \Big),
\end{align}
where the matrix identity $a(I+ba)^{-1} = (I+ab)^{-1}a$ is applied together with the matrix inversion lemma \cite{HJ_2007} and the matrix $M$ from (\ref{M}). By combining (\ref{pp1})  with (\ref{pipi}), (\ref{qdot}), (\ref{gamma}), it follows that
\begin{align}
\nonumber
    M\rd \dot{q}
    & =
    MM_0^{-1}
    \rd
    {\begin{bmatrix}
        p_1\\
        p_2
    \end{bmatrix}}
    =
        -
    Z^\rT
    \gamma \rd t
    +
    \rd
    {\begin{bmatrix}
        \pi_1\\
        \pi_2
    \end{bmatrix}}\\
\nonumber
    & =
        -
    Z^\rT
    (\kappa Z q+\phi Z\dot{q}) \rd t
    -
      (K_0q + F_0\dot{q})\rd t
        +
          N^\rT
        \rd W\\
\label{pp2}
    & =
    (-K q-F\dot{q}) \rd t
        +
          N^\rT
        \rd W,
\end{align}
where use is also made of the matrices $K$, $F$ from (\ref{K}), (\ref{F}).  A comparison of (\ref{pp2}) with (\ref{LSH2q}) shows that the interconnection being considered is indeed an LSH $(K,M,F,N)$-system,  whose matrices are given by (\ref{K})--(\ref{N}), and the momentum $p$ is related to the individual momenta $p_1$, $p_2$ as
\begin{equation}
\label{pMq}
    p
    =
    M\dot{q}
     =
    MM_0^{-1}
    {\begin{bmatrix}
        p_1\\
        p_2
    \end{bmatrix}}
    =
    (I_m + Z^\rT \mu ZM_0^{-1})
    {\begin{bmatrix}
        p_1\\
        p_2
    \end{bmatrix}},
\end{equation}
which elucidates the role of the inerter in the dynamic coupling of the systems.
\end{proof}

The matrices $K$, $M$, $F$ in (\ref{K})--(\ref{F}) depend affinely on the stiffness, inertance and damping matrices $\kappa$, $\mu$, $\phi$ of the coupling and are organised as a rank $m$ modification of the corresponding matrices (\ref{KMF0}). 
The relations (\ref{K})--(\ref{N}) can also be obtained by applying an appropriately modified formalism of Lagrangian mechanics, which takes into account  all the conservative and nonconservative elements of the system through their work-energy contributions.
In particular, the inerter in (\ref{f12}), (\ref{f21}) contributes the term $\frac{1}{2}\|\dot{y}_1 - \dot{y}_2\|_{\mu}^2 = \frac{1}{2}\|N_1 \dot{q}_1 - N_2\dot{q}_2\|_{\mu}^2 = \frac{1}{2}\|Z\dot{q}\|_{\mu}^2$ in the kinetic energy
$T = \frac{1}{2} \|\dot{q}_1\|_{M_1}^2 + \frac{1}{2} \|\dot{q}_2\|_{M_2}^2+\frac{1}{2}\|\dot{y}_1 - \dot{y}_2\|_{\mu}^2$  of the interconnected system (thus leading to (\ref{M})) due to the identity
\begin{align}
\nonumber
    (\dot{y}_1-\dot{y}_2)^\rT
    \mu \rd (\dot{y}_1-\dot{y}_2)
     =&
    \frac{1}{2}
    \rd
    \big(
    \|\dot{y}_1 - \dot{y}_2\|_{\mu}^2
    \big)\\
\label{Ito}
    & -
    \frac{1}{2}
    \Bra
        \mu,\,
        Z
        M^{-1}
        N^\rT
        \Sigma
        N
        M^{-1}
        Z^\rT
    \Ket
    \rd t,
\end{align}
whose left-hand side is the incremental work on the inerter. 
Here, $\Sigma$ is the diffusion matrix of the $\mR^{2m}$-valued process $W$ from (\ref{Wy}) in the sense of (\ref{dWdW}), so that $M^{-1}N^\rT\Sigma NM^{-1}$ is the diffusion matrix of $\dot{q}$. The Ito correction term on the second line  of (\ref{Ito}) is similar to that in (\ref{dH}).

From (\ref{K})--(\ref{F}), it follows  that if $\kappa\succcurlyeq 0$  and both LSH systems have positive definite stiffness and damping matrices $K_1$, $K_2$, $F_1$, $F_2$, then $K\succ 0$ and $F\succ 0$, which  makes the interconnection an internally stable LSH system. The state-space matrices of this $(K,M,F,N)$-system depend on the parameters of the constituent LSH systems and the coupling between them. Some of these parameters (those of them which can be varied) can be  chosen so as to improve the performance of the system dictated by a particular control objective.
In this approach, the first of the LSH systems can be regarded  as a plant, while the second of them plays the role of an analog (rather than digital) controller which implements feedback through the physical coupling with the plant without measurements as such. 

\section{\bf Mean square optimal control and first-order necessary conditions of optimality}\label{sec:msopt}

If the $(K,M,F,N)$-system, specified by (\ref{dXlin})--(\ref{C}), (\ref{K})--(\ref{N}) and arising from the interconnection in Fig.~\ref{fig:link}, is driven by a stationary Gaussian Ito process $W$ in (\ref{Wy}), then it has a unique invariant measure which is also Gaussian. If the process $W$ is modelled as the output of a stable  linear time-invariant shaping filter
\begin{equation}
\label{abcd}
    \rd \xi = a \xi \rd t + b \rd \ups,
    \qquad
    \rd W = c \xi\rd t  + d \rd \ups,
\end{equation}
driven by the standard Wiener process $\ups$ in (\ref{dW}), so that $\alpha = c\xi$ and $\beta = d$, then the incorporation of the internal state $\xi$ of the filter leads to an augmented system with the state-space realization $\Big({\small\begin{bmatrix}
  A & Bc\\ 0 & a
\end{bmatrix}}, {\small\begin{bmatrix}
  Bd \\ b
\end{bmatrix}}, {\small\begin{bmatrix}
  C & 0
\end{bmatrix}}\Big)$, whose non-Hamiltonian part can come only from the dynamics of $W$ in (\ref{abcd}), and the input is $\ups$. In particular, if $W$ itself is a standard Wiener process, then the system augmentation is not needed. In this case,  the solution $x$ of the SDE (\ref{dXlin}) is a diffusion process with an invariant zero-mean Gaussian probability distribution 
in $\mR^{4n}$ whose covariance matrix 
$    P :=
    \int_0^{+\infty} \re^{tA} BB^{\rT} \re^{tA^\rT} \rd t
$
is the controllability Gramian of the pair $(A,B)$ satisfying 
the algebraic Lyapunov equation (ALE)
\begin{equation}
\label{PALE}
    A P + P A^{\rT} + BB^{\rT} = 0.
\end{equation}
In the LQG control framework, the infinite-horizon performance of the interconnected system can be quantified in terms of the steady-state mean square cost
\begin{equation}
\label{Ups}
    \Ups
    :=
    \frac{1}{2}
    \bE(|\cC x|^2),
\end{equation}
which has to be minimised. Here, $\bE(\cdot)$ denotes expectation, and  $\cC$ is an appropriately dimensioned real matrix, which specifies the relative importance of the system variables. For example, if 
\begin{equation}
\label{cC}
    \cC
    :=
    {\begin{bmatrix}
    {\begin{bmatrix}
      \sqrt{\Pi_1} & 0
    \end{bmatrix}}
    M^{-1}
    {\begin{bmatrix}
      K & FM^{-1}
    \end{bmatrix}}\\
    {\begin{bmatrix}
      0 & \sqrt{\Pi_2}
    \end{bmatrix}}
    C
    \end{bmatrix}},
\end{equation}
where $C$ is given by (\ref{C}) and $\Pi_1$, $\Pi_2$ are positive definite matrices of orders $n$, $m$, then the cost functional in (\ref{Ups}) penalises the invariant mean value of the sum
\begin{equation}
\label{quad}
    |\cC x|^2 = \|(M^{-1})_{1\bullet}(Kq + F\dot{q})\|_{\Pi_1}^2 + \|y_2\|_{\Pi_2}^2
\end{equation}
of weighted Euclidean norms of the drift in the plant velocity $\dot{q}_1$ (with the drift corresponding to the acceleration) and the controller output $y_2$ (playing the role of an actuator signal). Here, $(M^{-1})_{1\bullet} \in \mR^{n\x 2n}$ is the upper block row of the matrix $M^{-1}$, so that $\dot{q}_1 = (M^{-1})_{1\bullet}p$ in accordance with (\ref{q}), (\ref{pMq}).  The minimisation of the acceleration in the mean square sense is relevant, for example,  to ride comfort improvement in vehicle suspension design \cite{S_2003} and other vibration  isolation problems. The specific structure of Ito processes can be taken into account in the optimal control setting  by considering more complicated cost functionals involving  the diffusion matrix $(M^{-1})_{1\bullet}N^\rT N(M^{-1})_{1\bullet}^\rT$ of $\dot{q}_1$ in addition to (\ref{quad}), which, however, is beyond the scope of this paper and will be discussed elsewhere.

The choice of a finite-dimensional parameter $\theta$, over which the mean square cost (\ref{Ups}) is minimised, depends on a particular application. For example,  if the inerter-spring-damper coupling is adjustable, then this suggests
\begin{equation}
\label{theta3}
    \theta := (\mu,\kappa,\phi).
\end{equation}
In general, $\theta$ can take values in an open subset $\Theta$ of a product of real matrix spaces endowed with a  Hilbert space structure with the direct-sum inner product. The minimisation of the mean square cost (\ref{Ups}) over $\theta \in \Theta$ can be carried out by using variational techniques. If the map $\theta \mapsto (K,M,F,N)$ is Frechet differentiable, the matrices $K$, $M$, $F$ are positive definite for every $\theta \in\Theta$, and the map $(K,M,F,N)\mapsto \cC$ is also Frechet differentiable (it is so for the dependence of $A$, $B$ on $K$, $M$, $F$, $N$ in (\ref{A}), (\ref{B})), then the  Frechet differentiability is inherited by the composite function
\begin{equation}
\label{comp}
    \theta \mapsto (K,M,F,N)\mapsto (A,B,\cC) \mapsto \Ups.
\end{equation}
Here, the map $(A,B,\cC) \mapsto \Ups$ is differentiable due to $A$ being Hurwitz.
This allows the first-order necessary conditions of optimality in the minimisation problem
\begin{equation}
\label{opt}
    \Ups\to \inf,
    \qquad
    \theta \in \Theta,
\end{equation}
to be obtained in the form $\d_\theta \Ups = 0$. The Frechet derivative $\d_\theta \Ups$ of the cost $\Ups$ in (\ref{Ups}) is also of interest from the infinitesimal perturbation analysis viewpoint since it quantifies the sensitivity of the system performance to the adjustable parameters. Its computation can be carried out by applying the chain rule to
the composite function (\ref{comp}):
\begin{align}
\label{chain1}
    \d_\theta \Ups
    & =
    \d_\theta K^\dagger
    (\d_K \Ups)
    +
    \d_\theta M^\dagger
    (\d_M \Ups)
    +
    \d_\theta F^\dagger
    (\d_F \Ups)
    +
    \d_\theta N^\dagger
    (\d_N \Ups),\\
\label{chain2}
    \d_K \Ups
    & =
    \d_K A^\dagger
    (\d_A \Ups)
    +
    \underbrace{\d_K B^\dagger}_0
    (\d_B \Ups)
    +
    \d_K \cC^\dagger
    (\d_\cC \Ups),\\
    \d_M \Ups
    & =
    \d_M A^\dagger
    (\d_A \Ups)
    +
    \underbrace{\d_M B^\dagger}_0
    (\d_B \Ups)
    +
    \d_M \cC^\dagger
    (\d_\cC \Ups),\\
\label{chain4}
    \d_F \Ups
    & =
    \d_F A^\dagger
    (\d_A \Ups)
    +
    \underbrace{\d_F B^\dagger}_0
    (\d_B \Ups)
    +
    \d_F \cC^\dagger
    (\d_\cC \Ups),\\
\label{chain5}
    \d_N \Ups
    & =
    \underbrace{\d_N A^\dagger}_{0}
    (\d_A \Ups)
    +
    \d_N B^\dagger
    (\d_B \Ups)
    +
    \d_N \cC^\dagger
    (\d_\cC \Ups),
\end{align}
where $(\cdot)^\dagger$ is the operator adjoint. The relations   (\ref{chain1})--(\ref{chain5}) are represented briefly as
\begin{equation}
\label{chain}
    \d_\theta \Ups
    =
    \d_\theta (K,M,F,N)^\dagger
    (\d_{K,M,F,N}(A,B,\cC)^\dagger
    (\d_{A,B,\cC}\Ups))
\end{equation}
in terms of operator extensions of the gradient vectors and Jacobian matrices (consisting of the partial Frechet derivatives of matrix-valued functions in matrix-valued variables).  The sparsity of the Jacobian matrix $\d_{K,M,F,N}(A,B,\cC)$ in (\ref{chain2})--(\ref{chain5}) follows from (\ref{A}), (\ref{B}), since $A$ does not depend on $N$, while $B$ depends only on $N$. By the results of \cite{BH_1998,SIG_1998,VP_2010b} on the differentiation of LQG costs,
\begin{equation}
\label{dUpsdABC}
    \d_A \Ups = \Gamma,
    \qquad
    \d_B \Ups = QB,
    \qquad
    \d_\cC \Ups = \cC P,
\end{equation}
where
\begin{equation}
\label{QP}
    \Gamma:= QP
\end{equation}
is the \emph{Hankelian} for the matrix triple $(A,B,\cC)$ (the eigenvalues of $\Gamma$ are the squares of  the Hankel singular values), associated with (\ref{PALE}) and the observability Gramian  $Q:= \int_0^{+\infty} \re^{tA^\rT} \cC^{\rT}\cC \re^{tA} \rd t$ of the pair $(A,\cC)$ satisfying the ALE
\begin{equation}
\label{QALE}
    A^\rT Q + Q A + \cC^{\rT}\cC = 0.
\end{equation}
The Hankelian  $\Gamma$ is split into four $(2n\x 2n)$-blocks $(\cdot)_{jk}$ in accordance with the partitioning  of the interconnected system variables in (\ref{HR}) into positions and momenta, while the matrices $B$, $\cC$ are split into  two blocks $(\cdot)_j$, $(\cdot)_k$, with $j,k=1,2$. 
Also, similarly to \cite{VP_2013a}, we will use a ``sandwich'' operator  $\[[[u,v\]]]$, specified by real matrices $u$, $v$ and acting on appropriately dimensioned real matrices $z$ as $\[[[u,v\]]](z):= uzv$.  Its adjoint is another such operator: $\[[[u,v\]]]^\dagger = \[[[u^\rT,v^\rT\]]]$. However, if it is restricted to the subspace $\mS$ of real symmetric matrices, its adjoint is the composition $(\[[[u,v\]]]\big|_\mS)^\dagger = \bS\[[[u^\rT,v^\rT\]]]$ with the symmetrizer $\bS(z) = \frac{1}{2}(z+z^\rT)$ of square matrices. At the same time, $(\[[[u,u^\rT\]]]\big|_\mS)^\dagger = \[[[u^\rT,u\]]]\big|_\mS$. These operators allow the Frechet derivatives of the matrices $A$, $B$ in (\ref{A}), (\ref{B}) with respect to $K$, $M$, $F$, $N$ to be represented as
\begin{align}
\label{diff1}
    \d_K A
    & =
    -
    \[[[
    {\begin{bmatrix}
      0\\
      I
    \end{bmatrix}},
    {\begin{bmatrix}
      I & 0
    \end{bmatrix}}
    \]]],
    \quad
    \d_M A
    =
    \[[[
    {\begin{bmatrix}
      -M^{-1}\\
      FM^{-1}
    \end{bmatrix}},
    {\begin{bmatrix}
      0 & M^{-1}
    \end{bmatrix}}
    \]]],\\
\label{diff2}
    \d_F A
    & =
    -
    \[[[
    {\begin{bmatrix}
      0\\
      I
    \end{bmatrix}},
    {\begin{bmatrix}
      0 & M^{-1}
    \end{bmatrix}}
    \]]],
    \quad
    \d_N B
    =
    \[[[
    {\begin{bmatrix}
      0\\
      I
    \end{bmatrix}},
    I
    \]]]
    \bT,
\end{align}
where the dimensions are omitted for brevity, and $\bT$ is the matrix transpose operator. 
Also, in the case of (\ref{theta3}),  only the following entries of the corresponding Jacobian matrix $\d_{\mu,\kappa,\phi}(K,M,F,N)$ in (\ref{chain}) are nontrivial in view of (\ref{K})--(\ref{N}):
\begin{equation}
\label{ZZ}
    \d_\mu M = \d_\kappa K = \d_\phi F = \[[[ Z^\rT, Z\]]],
\end{equation}
where the matrix $Z$ is defined in (\ref{dom}). Therefore, a combination of (\ref{ZZ}) with (\ref{chain1})--(\ref{chain4}), (\ref{dUpsdABC})--(\ref{diff2}) yields
\begin{align}
\nonumber
    \d_\mu \Ups
     & =
    \d_\mu M^\dagger
    (\d_M \Ups)
    =
    Z\d_M \Ups Z^\rT\\
\label{der1}
    & =
    Z
    (
        M^{-1}
        \bS(F\Gamma_{22}-\Gamma_{12})
        M^{-1}
        +
        \d_M\cC^\dagger(\cC P)
    )
    Z^\rT,\\
\nonumber
    \d_\kappa \Ups
     & =
    \d_\kappa K^\dagger
    (\d_K \Ups)
    =
    Z\d_K \Ups Z^\rT\\
\label{der2}
    & =
    Z
    (
        -\bS(\Gamma_{21}) + \d_K\cC^\dagger(\cC P)
    )
    Z^\rT,\\
\nonumber
    \d_\phi \Ups
     & =
    \d_\phi F^\dagger
    (\d_F \Ups)
    =
    Z\d_F \Ups Z^\rT\\
\label{der3}
    & =
    Z
    (
        -\bS(\Gamma_{22}M^{-1}) + \d_F\cC^\dagger(\cC P)
    )
    Z^\rT.
\end{align}
Here, the Frechet  derivatives  of the matrix $\cC$ in (\ref{cC}) with respect to $K$, $M$, $F$ are computed in a similar fashion, except that $\d_M\cC$ involves the sum of two sandwich operators (we omit these calculations for brevity). In view of (\ref{der1})--(\ref{der3}), the first-order necessary conditions of optimality $\d_\mu \Ups =0$, $\d_\kappa \Ups =0$, $\d_\phi \Ups =0$ for the problem (\ref{opt}) are organised as a set of nonlinear matrix algebraic equations whose solution will be considered  elsewhere. A similar variational approach can be applied to optimizing the interconnection performance  not only over the parameters of the coupling but also  with respect to those of the LSH controller itself.

\section{\bf Conclusion}\label{sec:conc}

We have discussed a class of LSH systems, governed by linear SDEs and specified by stiffness, mass, damping and coupling matrices. A multivariable stochastic version of the inerter-spring-damper couplings has been considered for such systems. We have discussed a mean square optimal control problem  for this interconnection of LSH systems  with respect to the 
coupling parameters, and outlined first-order necessary conditions of optimality using variational techniques developed previously for constrained LQG control settings. These conditions employ Frechet derivatives of the quadratic cost functional,  which can be  used for numerical solution of the optimization problem via Newton or gradient descent methods and for infinitesimal perturbation analysis of such systems.
%

\end{document}